\documentclass[conference,letterpaper]{IEEEtran}

\IEEEoverridecommandlockouts

\usepackage[utf8]{inputenc} 
\usepackage[T1]{fontenc}
\usepackage{url}
\usepackage{ifthen}
\usepackage{cite}
\usepackage{dsfont}
\usepackage[cmex10]{amsmath} 

\usepackage{physics}
\usepackage{amssymb}
\usepackage{amsthm}
\usepackage{xcolor}
\usepackage{ulem}
\usepackage{mathtools}
\usepackage{hyperref}
\usepackage{cleveref}
\usepackage{enumitem}

\newtheorem{ithm}{Theorem}
\newtheorem{theorem}{Theorem}
\newtheorem{definition}{Definition}

\newtheorem{lemma}[theorem]{Lemma}

\newtheorem{prop}[ithm]{Proposition}

\newcommand{\ve}{\varepsilon}

\newcommand{\mbb}{\mathbb}

\newcommand{\tinyspace}{\mspace{1mu}}

\renewcommand{\norm}[1]{\lVert\tinyspace #1 \tinyspace\rVert}

\newcommand{\cE}{\mathcal{E}}
\newcommand{\cF}{\mathcal{F}}

\newcommand{\cH}{\mathcal{H}}

\newcommand{\cM}{\mathcal{M}}

\newcommand{\contractsuppi}{\sup_{\substack{\rho \in \mathcal{D(H)} \\ 0 < \mathbb{D}_f(\rho \Vert \pi) < \infty}}}

\newcommand{\mixedterm}{\sigma^{-\frac{1}{2}} \rho \sigma^{-\frac{1}{2}}}

\newcommand{\Petz}{\mathrm{Petz}}

\begin{document}
\title{Tight Contraction Rates for Primitive Channels under Quantum $f$-Divergences\thanks{MST acknowledges support from the NUS REx Grant. MT and IG were supported by the Ministry of Education, Singapore, through grant T2EP20124-0005. IG is supported by the National Research Foundation, Singapore under the NRF Postdoctoral award. MT is also supported by the NRF Investigatorship award (NRF-NRFI10-2024-0006).}}


\author{
\IEEEauthorblockN{Matthew Simon Tan}
\IEEEauthorblockA{\textit{School of Computing} \\
\textit{National University of Singapore}\\
Singapore 117417, Singapore \\
e0726673@u.nus.edu} 
\and
\IEEEauthorblockN{Marco Tomamichel}
\IEEEauthorblockA{\textit{Dept.~of Electrical and Computer Engineering \&} \\
\textit{Centre for Quantum Technologies} \\
\textit{National University of Singapore}\\
Singapore 117543, Singapore}
\and
\IEEEauthorblockN{Ian George}
\IEEEauthorblockA{\textit{Centre for Quantum Technologies} \\
\textit{National University of Singapore}\\
Singapore 117543, Singapore \\
qit.george@gmail.com}
}

\maketitle


\begin{abstract}
  Data-processing inequalities capture the phenomenon that two probability distributions can only become less distinguishable under any common post-processing. For more fine-grained inequalities, one turns to strong data-processing inequality (SDPI) constants, which give the strongest inequalities for a given channel and reference state for a fixed measure of distinguishability. These quantities have been used to quantify the rate at which time-homogeneous Markov chains contract towards a fixed point both in the classical and quantum setting. In this work, we establish that quantum $f$-divergences satisfy a local reverse Pinsker inequality, which implies the asymptotic contraction rate of a primitive channel to its stationary state is upper bounded by the SDPI constant of any non-commutative $\chi^2$-divergence. Using quantum-detailed balance, we establish a sufficient condition for these bounds to be tight. Finally, we apply these results to Petz, Matsumoto, and Hirche-Tomamichel $f$-divergences, establishing new and strengthening previously known results.
\end{abstract}

\section{Introduction}
 In information theory the \textit{data-processing inequality} (DPI) of a divergence formalizes the idea that two states undergoing common post-processing can only become less distinguishable. In classical information theory, divergences that satisfy the data-processing inequality can be parameterized by convex functions. These are known as Csiszár $f$-divergences. Given a divergence $D_f$, a classical channel $W$, and an input probability distribution $p$, one can compute the corresponding strong data-processing inequality (SDPI) constant $\eta_f(W,p)$. The study of these  constants, initiated in \cite{ahlswede1976spreading}, give us the strongest possible data-processing inequalities we can obtain given only the knowledge of a fixed channel and reference state. 
 
 In quantum information theory, the non-commutativity of elements in the state space mean there are multiple distinct generalizations of $f$-divergences that each satisfy the DPI, see e.g.~\cite{hiai2011quantum,hiai2019quantum,hirche2024quantum}. Their SDPI constants have similarly been studied, see e.g.~\cite{lesniewski1999monotone, hiai2016contraction,hirche2022contraction, beigi2025some,george2025unified,iyer2025quantum} and references therein. While quantum $f$-divergences are perhaps initially defined merely in hopes of extending the DPI to quantum theory, they also can extend mathematical properties and operational interpretations of the classical $f$-divergences. Establishing such extensions allows us to both better understand the mathematics used in quantum information theory and to separate the relevance of different quantum $f$-divergences. Classically, the contraction rate of time-homogeneous Markov chains is upper-bounded by $\eta_{\chi^2}(W, p)$, the smallest SPDI constant \cite{makur2020comparison,george2024divergence}. Thus regardless of the chosen $f$-divergence, the contraction rate is "as fast as possible." In this work, we study quantum $f$-divergence SDPI constants and their applications to time-homogeneous Markov chains. In particular, we establish the following: 

 \begin{theorem}\label{thm:main-result}
    Let $f\in \mathcal{F}_{\mathrm{Pin}}, g \in \mathcal{M}_{\mathrm{st}}$, and $\mathcal{E}$ be a primitive channel with a fixed point $\pi >0$. Assume further $f$ is thrice continuously differentiable in an open neighborhood of $1$. Then we have 
    \begin{equation}
        \lim_{n \to \infty}\eta_{\, \mbb{D}_{f}}(\mathcal{E}^{\circ n}, \pi) ^{\frac{1}{n}}\leq  \eta_{\chi^2_{g}}(\mathcal{E},\pi)
    \end{equation}
    Moreover, if $\mathbb{D}_f$ is locally $\chi^{2}_{g}$, and $\mathcal{E}$ satisfies $g$-detailed balance with respect to $\pi$, then we have equality.
\end{theorem}
\noindent We stress the appeal of this theorem is that one would expect most quantum $f$-divergences to be locally some $\chi^{2}_{g}$ as this is true in the commutative setting \cite{polyanskiy2025information}, and so this theorem can be applied as soon as the second-order behavior of the divergence is determined. We present examples of this in Section \ref{sec:applications}.

The special case of Theorem \ref{thm:main-result} has been established classically in \cite{makur2020comparison,george2024divergence} for irreducible, aperiodic time-homogeneous Markov chains. Analogous results have been established for the Hirche-Tomamichel (HT) $f$-divergences in \cite{george2025unified}, and the Petz $f$-divergences in \cite{george2024divergence}. In all preceding quantum cases, tightness was only established in the case where the input states commute, arising from the classical tightness condition in \cite{makur2020comparison}. We provide proofs in terms of general quantum $f$-divergence, and thereby recover the aforementioned results as special cases. Using the family of quantum detailed-balance conditions defined in \cite{temme2010chi}, we also offer a sufficient condition for tightness that extends to the non-commutative setting. Applying these results to Petz, Matsumoto, and Hirche-Tomamichel $f$-divergences, establishing new and strengthening previously known results.

 The rest of the proceeding proves Theorem \ref{thm:main-result}. Omitted proofs may be found in \cite{Tan-2026}.

\section{Preliminaries}
Let $\mathcal{H}$ denote a finite-dimensional complex Hilbert space of fixed dimension. We let $\mathcal{D}(\mathcal{H})$ denote the set of density matrices on this Hilbert space, which are positive semi-definite operators of trace $1$. Quantum channels are maps $\mathcal{E}:\mathcal{D}(\mathcal{H}) \to \mathcal{D}(\mathcal{H})$, which are completely positive and trace-preserving. Each of these maps can naturally be extended uniquely by linearity to a linear operator on the set of bounded operators on $\mathcal{H}$, denoted by $\mathcal{B}(\mathcal{H})$. Given a channel $\cE$, we define $\cE^{\circ n} = \bigcirc_{i \in [n]} \cE$. To avoid issues with support, we assume $\rho, \sigma$ are full rank for simplicity when it is not specified. $\norm{\cdot }_{p}$ for $p \in [1,\infty]$ denote the Schatten $p$-norms. When not specified, all inner products are the Hilbert-Schmidt inner product. Finally, we use $x \lesssim_{f,g,\pi} y$ to denote there exists some constant $c>0$ depending only on $f,g,\pi$ such that $x \leq c
\cdot y$. Note we allow this constant to be dimension dependent as well, and when there is no subscript, this constant depends only on the dimension of the Hilbert space. These inequalities also compose, with the dependency of constants accumulating, so if $A \lesssim_{f} B$ and $B \lesssim_{g} C$ we have $A \lesssim_{f,g} C$. We begin by reviewing several important operators and classes of functions that will be used throughout this proceeding. A more detailed introduction of the quantities discussed can be found in \cite{george2025unified}.

We begin by introducing quantum $f$-divergences generally as distinguishability measures that generalize the Csiszár $f$-divergences \cite{petz1985quasi, hirche2024quantum,matsumoto2015new}. 
\begin{definition}
    Let $f:(0,\infty) \to (0,\infty)$ be twice continuously differentiable convex function with $f(1) = f'(1) =0$.  We say a map $\mathbb{D}_f: \mathcal{D(H)} \times \mathcal{D(H)} \to \mathbb{R}_{\geq 0}$ is a quantum $f$-divergence if it satisfies the following properties: 
    \begin{enumerate}
        \item (DPI) $\mathbb{D}_f(\mathcal{E}(\rho) \Vert \mathcal{E}(\sigma)) \leq \mathbb{D}_f(\rho \Vert \sigma)$ for all $\rho,\sigma \in \mathcal{D(H)}$ and $\mathcal{E}$ that is CPTP. 
        \item (Classical Consistency) $\mathbb{D}_f(\rho \Vert \sigma) = D_f(p \Vert q)$ if $[\rho,\sigma] =0$, where $p,q$ are the classical probability distributions given by the spectrum of $\rho,\sigma$ respectively. 
    \end{enumerate}
    The family of functions $f$ that satisfy the above conditions are denoted with $\mathcal{F}$.
\end{definition}

\begin{definition}
    We say $f \in \cF_{\mathrm{Pin}}$ if $f \in \cF$ and there exists some constant $C_f>0$ depending only on dimension and $f$ such that $D_f(p \Vert  q) \geq C_f \norm{p - q}_1^2$ for all probability distributions $p$ and $q$, i.e. satisfies a Pinsker inequality. By Corollary 62 of \cite{george2024divergence}, if $f \in \cF_{\mathrm{Pin}}$, then for all $\rho,\sigma \in \mathcal{D(H)}$,
    \begin{equation}
        \mathbb{D}_f(\rho  \Vert  \sigma ) \geq C_f \norm{\rho - \sigma}_1^2 \ . 
    \end{equation}
\end{definition}

Strong data-processing inequality (SDPI) constants can then be defined in this more general framework as follows:

\begin{definition}[SDPI Constant]
    Fix an $f \in \mathcal{F}$. For a CPTP map $\mathcal{E}$ and a state $\sigma \in \mathcal{D(H)}$, we define the SDPI constant of a channel $\mathcal{E}$ with respect to a state $\sigma$ by 
    \begin{equation}
        \eta_{\mbb{D}_{f}}(\mathcal{E}, \sigma) = \sup_{\substack{\rho \in \mathcal{D(H)} \\ 0 < \mathbb{D}_f(\rho\Vert \sigma) < \infty}} \frac{\mathbb{D}_f(\mathcal{E}(\rho)\Vert \mathcal{E}(\sigma))}{\mathbb{D}_f(\rho\Vert \sigma)} \ . 
    \end{equation}
     When we consider the SDPI constant of a specific $\chi^2_g$-divergence, we will use $\eta_{\chi^2_g}$. 
\end{definition}

\begin{definition}
    Let $\mathcal{H}$ be a Hilbert space, and $P,Q \geq  0$. The relative modular operator is the super-operator $\Delta_{P,Q}: \mathcal{B(H)} \to \mathcal{B(H)}$ is defined by 
    \begin{equation}
        \Delta_{P,Q} = L_{P}R_{Q^{-1}}
    \end{equation}
    where $L,R$ are the left and right multiplication super-operators respectively defined by $L_P(X) = PX$ and $R_{Q}(X) = XQ$. 
\end{definition}

\begin{definition}
    A function $f:\mathbb{R}_{>0} \to \mathbb{R}$ is 
    \begin{enumerate}
        \item (Operator Monotone Decreasing) if  $f(\rho) \leq f(\sigma)$ for all PSD $\rho,\sigma$ such that $\rho \leq \sigma$. 
        \item (Normalized) if  $f(1)=1$. 
        \item (Symmetry-Inducing) if $xf(x^{-1}) = f(x)$ for $x \in \mathbb{R}_{>0}$. 
    \end{enumerate}
    We say that $f$ is standard monotone if it satisfies properties $1,2$ and $3$, and denote this family of functions by $\mathcal{M}_{\mathrm{st}}$. 
\end{definition} 

Note that operator continuity is sufficient to imply continuity (one can see this by considering the integral decomposition of an operator monotone function), and the above three conditions imply $g(x)>0$ for $x>0$. This implication will see heavy use in the subsequent proofs. These standard monotone functions define a family of non-commutative generalizations of $\chi^2$-divergences, as identified in \cite{temme2010chi}. This is done by using the family of standard monotone functions $\mathcal{M}_{st}$ to parameterize the various ways one can perform non-commutative division by $\sigma$. 

\begin{definition}
    Let $0<\sigma \in \mathcal{D(H)}$ and $f\in \mathcal{M}_{\mathrm{st}}$. Then 
    \begin{equation}
    \Omega_{\sigma}^{f} := R^{-1}_{\sigma}f(\Delta_{\sigma,\sigma}),
\end{equation}
is a quantum inversion of $\sigma$. 
\end{definition}

Each inversion gives rise to a quantum $\chi^2$-divergence in the following manner.

\begin{definition}\cite{temme2010chi}
    Let $f \in \mathcal{M}_{\mathrm{st}}$ and $\rho,\sigma \in \mathcal{D(H)}$. Then the quantum $\chi^2$-divergence defined by $f,\sigma$ is the following map from $\mathcal{D(H)} \times \mathcal{D(H)} \to \mathbb{R}_{\geq 0}$:
    \begin{equation}
      \chi^2_f(\rho\Vert \sigma) = \langle \rho -\sigma, \Omega_{\sigma}^{f}(\rho-\sigma)\rangle \ . 
    \end{equation}
\end{definition}

An important example of the quantum $\chi^2$-divergences is the maximal $\chi^{2}$-divergence:
\begin{equation}
    \chi^2_{\max}(\rho \Vert \sigma) = \Tr[\sigma^{-1}(\rho-\sigma)^2] \ . 
\end{equation}
The name is justified by the known fact \cite{hiai2016contraction} that for all $f \in \cM_{st}$,
\begin{equation}
    \chi^2_{f} \leq \chi^{2}_{\max} \ .
\end{equation}

\section{General contraction rate upper bounds}
In this section, we prove the upper bounds given in Theorem \ref{thm:main-result}. An analogue of time-homogeneous, irreducible aperiodic Markov chains is the following: 
\begin{definition}
    Let $\mathcal{E}$ be a quantum channel. We say that $\mathcal{E}$ is primitive if there exists $n \in \mathbb{N}$ for which 
    \begin{equation}
        \mathcal{E}^{\circ n}(\rho) >0 \quad \forall \rho \in \mathcal{D(H)} \ . 
    \end{equation}
   This condition is equivalent to $\mathcal{E}$ having a unique full rank fixed-point (see \cite[Theorem 6.7]{wolf2012quantum}). That is, $\mathcal{E}^{\circ n}(\rho) \to \pi$ for some $\pi>0$ as $n \to \infty$, for all $\rho \in \mathcal{D(H)}$.
\end{definition}

A crucial fact we require is that this convergence is uniform in $\rho$. 

\begin{lemma}[Uniform convergence in $\rho$.]\label{lemma:uniform}
    Let $\mathcal{E}$ be a primitive channel with full rank stationary distribution $\pi$. Then for any $\varepsilon>0$, there exists $N_{\varepsilon} \in \mathbb{N}$ such that for all $n \geq N_{\varepsilon}$, we have 
    \begin{equation}
        \norm{\mathcal{E}^{\circ n}(\rho)- \pi}_{\infty} < \varepsilon
    \end{equation} for all $\rho \in \mathcal{D(H)}$. 
\end{lemma}

To bound the rate of contraction measured under quantum $f$-divergences by $\chi^{2}_{g}$ SDPI constants, we will need to establish quantum $f$-divergences satisfy a local reverse Pinsker inequality.

To establish quantum $f$-divergences satisfy this property, we use the following two technical lemmata. 
\begin{lemma}\label{lemma:eigenvalues} 
Let $\rho,\sigma \in \mathcal{D(H)}$, and $\sigma>0$. Then 

\begin{equation}
    \norm{\sigma^{-\frac{1}{2}} \rho \sigma^{-\frac{1}{2}} - \mathds{1}}_{\infty}  \leq \frac{1}{\lambda_{\min}(\sigma)} \norm{\rho-\sigma}_{\infty} \ .
\end{equation}
\end{lemma}

\begin{lemma}\cite{lanier2025classicalquantumexplicitclassical}\label{theorem:generalreversepinsker}
    For any $\rho, \sigma \in \mathcal{D(H)}$ satisfying the condition
\begin{equation}
|\rho - \sigma| \le \rho + \sigma,
\end{equation}
we have
\begin{equation}\label{eq:q-rev-pinsker}
\mathbb{D}_f(\rho \| \sigma) \leq \frac{\|\rho - \sigma\|_1}{2} \left( \frac{f(m)}{1 - m} + \frac{f(M)}{M - 1} \right).
\end{equation}
where $m = \lambda_{\min}(\sigma^{-\frac{1}{2}} \rho \sigma^{-\frac{1}{2}})$ and $M = \lambda_{\max}(\sigma^{-\frac{1}{2}} \rho \sigma^{-\frac{1}{2}})$ are the smallest and largest eigenvalues of $\sigma^{-\frac{1}{2}} \rho \sigma^{-\frac{1}{2}}$, respectively.
\end{lemma}

We now establish the wanted local reverse Pinsker inequality property.
\begin{prop}\label{prop:revpinsker}
    Let $\rho,\sigma \in \mathcal{D(H)}$ with $\sigma >0$. \ If $f \in \mathcal{F}$, and there is some $\varepsilon \in (0,1)$ for which $f$ is thrice continuously differentiable on $(1-\varepsilon,1+\varepsilon)$, and $\norm{\rho - \sigma}_{\infty} < \frac{\varepsilon \cdot \lambda_{\min}(\sigma)}{2}$, then 
    \begin{equation}
        \mathbb{D}_f(\rho  \Vert  \sigma) \lesssim_{f,\sigma}  \norm{\rho - \sigma}^{2}_{1} 
    \end{equation}
    where $m = \lambda_{\min}(\sigma^{-\frac{1}{2}} \rho \sigma^{-\frac{1}{2}})$ and $M = \lambda_{\max}(\sigma^{-\frac{1}{2}} \rho \sigma^{-\frac{1}{2}})$
\end{prop}
\begin{proof}
     Assume $\rho \neq \sigma$, otherwise the bound is trivial. It follows $m <1<M$. Our goal is now to bound the term on RHS of \eqref{eq:q-rev-pinsker} that depend on $f$, $m$, and $M$ in terms of $\Vert \rho - \sigma \Vert_{1}$. To do this, we first Taylor expand $f$ at $1$. Since $f(1) = f'(1) = 0$, we have 
    \begin{equation}
    \begin{aligned}
        \frac{f(m)}{1-m}  = \frac{f''(1)}{2}(1-m) - \frac{f'''(\xi_m)}{6} (1-m)^2
    \end{aligned}
    \end{equation}
    where $\xi_m \in (m,1)$. Similarly, we have 
    \begin{equation}
    \begin{aligned}
        \frac{f(M)}{M-1} = \frac{f''(1)}{2}(M-1) - \frac{f'''(\xi_M)}{6} (M-1)^2
    \end{aligned}
    \end{equation}
    where $\xi_M \in (1,M)$. 
    Thus we have 
    \begin{equation}\label{eq:local-reverse-pinsker-step-1}
    \begin{aligned}
        &\left|\frac{f(m)}{1-m} + \frac{f(M)}{M-1}\right|  \\ & \leq 
        \frac{f''(1)}{2}(M-m) + \left|\frac{f'''(\xi_M)}{6} \right| (M-1)^2 \\& \hspace{5mm} + \left| \frac{f'''(\xi_m)}{6}\right| (m-1)^2 \ .
    \end{aligned}
    \end{equation}
    
    We now bound each element of the sum in \eqref{eq:local-reverse-pinsker-step-1} in terms of $\Vert \rho - \sigma \Vert_{1}$. To do this, note that by definition of the $\sup$ norm,
    \begin{align}\label{equation:supnorm}
        \max\{M-1,1-m\} \leq \norm{\sigma^{-\frac{1}{2}} \rho \sigma^{-\frac{1}{2}} - \mathds{1}}_{\infty} \ . 
    \end{align}
    Thus, \begin{equation}\label{eq:local-reverse-pinsker-step-2}
    \begin{aligned}
    M-m = M-1 + 1-m &\leq 2\norm{\sigma^{-\frac{1}{2}} \rho \sigma^{-\frac{1}{2}} - \mathds{1}}_{\infty} \\
    & \lesssim_{\sigma} \norm{\rho-\sigma}_{\infty} \\ 
    & \lesssim \norm{\rho - \sigma}_1 \ , 
    \end{aligned}
    \end{equation}
    where we applied Proposition \ref{prop:revpinsker} and that all finite dimensional norms are equivalent. By the same argument,
    \begin{align} 
        \max\{M-1,1-m\} \lesssim_{f,\sigma} \norm{\rho - \sigma}_{1} \label{eq:local-reverse-pinsker-step-3} \ . 
    \end{align}

    
    Furthermore, since 
    $\norm{\rho - \sigma}_{\infty} < \frac{\varepsilon \cdot \lambda_{\min}(\sigma)}{2}$. So we have 
    \begin{equation}
        \xi_m,\xi_M \in (m,M) \subset [1-\frac{\varepsilon}{2}, 1 + \frac{\varepsilon}{2}]
    \end{equation}
    by applying Lemma \ref{lemma:eigenvalues} to \eqref{equation:supnorm}.
    Defining $C_{f,\sigma} \coloneqq \max_{x \in [1-\frac{\varepsilon}{2}, 1 + \frac{\varepsilon}{2}]} \vert f'''(x) \vert$, we conclude from the above
    \begin{equation*}
        \begin{aligned}
             \left|\frac{f(m)}{1-m} + \frac{f(M)}{M-1}\right| \lesssim_{f,\sigma} \norm{\rho - \sigma} + (M-1)^2 + (m-1)^2
        \end{aligned}
    \end{equation*}
    Combining this with Theorem \ref{theorem:generalreversepinsker} and then \eqref{eq:local-reverse-pinsker-step-2}, we have 
    \begin{align}
        \mathbb{D}_f(\rho \Vert \sigma) &\lesssim_{f,\sigma} \norm{\rho - \sigma}_1^{2} + (\norm{\rho - \sigma}_1)\left[(M-1)^2 + (m-1)^2 \right] \notag \\
        &\lesssim_{f,\sigma} \norm{\rho - \sigma}_1^{2} + \norm{\rho - \sigma}_1^3 \ . \label{eq:local-reverse-pinsker-step}
    \end{align}
    Since $\sigma>0$ and our Hilbert space dimension is at least two, we have that $\lambda_{\min}(\sigma) \leq  \frac{1}{\vert \cH \vert}$, and thus 
    \begin{equation*}
        \begin{aligned}
            \norm{\rho - \sigma}_1 & \leq \vert \cH \vert \cdot \norm{\rho - \sigma}_{\infty} < \frac{\vert \cH \vert}{2} \cdot \lambda_{\min}(\sigma) \leq \frac{1}{2} \ .
        \end{aligned}
    \end{equation*}
    So $\Vert \rho - \sigma \Vert^{3}_{1} < \Vert \rho - \sigma \Vert^{2}_{1}$,
    and the proposition follows.
\end{proof}

We also need to lower bound $\mathbb{D}_f$ by any $\chi^2_g$ divergence, which can be done as follows.  
\begin{prop}\label{lowerbound}
    Let $f \in \mathcal{F}_{\mathrm{Pin}}$, $g \in \mathcal{M}_{\mathrm{st}}$, and let $\mathbb{D}_f$ be a quantum $f$-divergence. Then for $\rho,\sigma \in \mathcal{D(H)}, \rho \ll \sigma$, we have 
    \begin{equation}
       \chi^2_{g}(\rho \Vert \sigma) \lesssim_{g,\sigma} \mathbb{D}_f(\rho  \Vert  \sigma)
    \end{equation}
    and we also have 
    \begin{equation}
        \norm{\rho - \sigma}_{1}^2 \lesssim_{g,\sigma} \chi^2_g(\rho \Vert \sigma) 
    \end{equation}
\end{prop}

This, combined with the fact that $\mathcal{E}^{\circ n}(\rho) \to \sigma$ uniformly in $\rho$, is sufficient to prove a generalization of the classical result. 

\begin{theorem}\label{thm:upper-bound}
     Let  $f \in \mathcal{F}_{\mathrm{Pin}}, g \in \mathcal{M}_{st}$. Let $\mathbb{D}_f$ be a quantum f-divergence. Let $\mathcal{E}$ be a primitive channel with fixed point $\pi$, and additionally assume $f$ is thrice continuously differentiable in an open neighborhood of $1$. Then we have 
    \begin{equation}
        \lim_{n \to \infty}\eta_{\mbb{D}_{f}}(\mathcal{E}^{\circ n}, \pi)^{1/n} \leq \eta_{\chi^2_{g}}(\mathcal{E},\pi) \ . 
    \end{equation}
\end{theorem}
\begin{proof}
    Take sufficiently large $N$ such that $\norm{\mathcal{E}^n(\rho)-\pi}_{\infty} < \frac{\lambda_{\min}(\pi)}{2}$ for all $\rho$ (note that such an $N$ exists due to the uniform convergence established in Lemma \ref{lemma:uniform}). Then for $n \geq N$, we have the following 
    \begin{equation}
    \begin{aligned}
        \eta_f(\mathcal{E}^{\circ n},\pi) & = \sup_{\substack{\rho \in \mathcal{D(H)} \\  0 < \mathbb{D}_f(\rho\Vert \pi) < \infty}} \frac{\mathbb{D}_f(\mathcal{E}^{\circ n}(\rho) \Vert  \pi)}{\mathbb{D}_f(\rho \Vert \pi)} \\ 
        & \lesssim_{f,g,\pi} \sup_{\substack{\rho \in \mathcal{D(H)} \\   0 < \mathbb{D}_f(\rho\Vert \pi) < \infty} }\frac{\mathbb{D}_f(\mathcal{E}^{\circ n}(\rho) \Vert  \pi)}{\chi^2_g(\rho \Vert \pi)} \\ 
        & \lesssim_{f,\pi} \contractsuppi \frac{\norm{\mathcal{E}^{\circ n}\rho - \pi}_1^{2}}{\chi^2_g(\rho  \Vert  \pi)}  \\ 
        & \lesssim_{g,\pi} \sup_{\substack{\rho \in \mathcal{D(H)} \\  0 < \mathbb{D}_f(\rho\Vert \pi) < \infty}}\frac{\chi^2_g(\mathcal{E}^{\circ n}(\rho) \Vert  \pi)}{\chi^2_g(\rho \Vert \pi)}  \\
         & = \eta_{\chi^2_g}(\mathcal{E}^{\circ n}, \pi)   \\
         &  \leq   \eta_{\chi^2_g}(\mathcal{E}, \pi)^{n} ,
    \end{aligned}
    \end{equation}
    where the first inequality lower bounds the denominator with Proposition \ref{lowerbound}, the second uses Proposition \ref{prop:revpinsker} to bound the numerator along with the fact that $\pi$ is a fixed point of $\mathcal{E}$, the third inequality uses Proposition \ref{lowerbound} and the final inequality holds by sub-multiplicativity of the SDPI constants \cite[Prop. 64]{george2024divergence} and the fact that $\pi$ is a fixed point of $\mathcal{E}$.
    
    Since the constant are fixed and independent of $n$, taking the $\frac{1}{n}$ powers on both sides and sending $n \to \infty$ proves the claim. 
\end{proof}

\section{Tightness of the upper bound}
In this section, we establish the tightness condition given in Theorem \ref{thm:main-result}. Classically, a sufficient condition for when the upper bound is tight is for the Markov chain given by $W$ to be reversible \cite{makur2020comparison}. The proof relies on making use of the fact that the transition matrix of a reversible time-homogeneous Markov chain satisfies the detailed-balance equation, and the following well-known property of SDPI constants: 
\begin{prop}\cite{raginsky2016strong,polyanskiy2017strong}
    If $f$ is twice differentiable at unity with $f''(1)>0$, then we have 
    \begin{equation}\label{eq:cl-SDPI-constant-ordering}
        \eta_{\chi^2}(P_{Y|X},P_X) \leq \eta_{f}(P_{Y|X},P_X)
    \end{equation}
\end{prop}
Therefore, one possible approach to generalizing the classical tightness conditions is to use the non-commutative generalizations of these two conditions. In this section, we briefly review these conditions, and show that they give a sufficient condition where the upper bound is tight.

To begin, we identify conditions that imply the SDPI constants inequality \eqref{eq:cl-SDPI-constant-ordering} holds. As was shown for HT divergences \cite{hirche2024quantum}, this can be seen as a manifestation of the fact that the local behaviour of an $f$-divergence is given by a $\chi^2$-divergence. We define this condition as follows: 
\begin{definition}
    Let $\mathbb{D}_f$ be a quantum $f$-divergence. We say that $\mathbb{D}_f$ is locally $\chi^2_{\kappa_f}$ if 
    \begin{equation}
        \lim_{\lambda \to 0}\frac{1}{\lambda^2}\mathbb{D}_f(\lambda \rho + (1-\lambda)\sigma  \Vert  \sigma) = \frac{f''(1)}{2}\chi^2_{\kappa_f}(\rho  \Vert  \sigma).
    \end{equation}
\end{definition}

We parameterize $\kappa_{f}$ in terms of $f$, since there are various non-commutative $\chi^2$-divergences and we do not know which coincides with the local behavior of the chosen quantum $f$-divergence $\mathbb{D}_f$ a priori. A direct generalization of the proof in \cite{hirche2024quantum} establishes the following.

\begin{lemma}\label{lemma:local}
    Suppose $\mathbb{D}_f$ is locally $\chi^2_{\kappa_f}$. Then $\eta_{\mbb{D}_{f}}(\mathcal{E},\sigma) \geq \eta_{\chi^2_{\kappa_f}}(\mathcal{E},\sigma)$.
\end{lemma}

Next, we review a family of non-commutative generalizations of detailed-balance  \cite{temme2010chi}.

\begin{definition}[Quantum Detailed Balance] \cite{temme2010chi}
    Let $f \in \mathcal{M}_{st}$, $\sigma \in \mathcal{D}(\mathcal{H})$, and let $\mathcal{E}$ be a super-operator. Then we say that $\mathcal{E}$ satisfies $f$-detailed balance with respect to $\sigma$ if 
    \begin{equation}\label{eq:DB-equation}
        (\Omega_{\sigma}^{f})^{-1} \circ \mathcal{E}^{*} = \mathcal{E} \circ (\Omega_{\sigma}^{f})^{-1}
    \end{equation}
\end{definition}
\noindent It should be noted that some quantum-detailed balanced conditions are strong enough to imply the others. In particular, the following result was established in \cite{carlen2017gradient}. 
\begin{theorem}\cite{carlen2017gradient}
    Let $\mathcal{E}$ be a hermitian preserving super-operator, $\sigma >0$. If $\mathcal{E}$ satisfies detailed balance with respect to the constant function $f(t) =1$ and $\sigma$, then $\mathcal{E}$ satisfies $g$ detailed balance for all $g \in \mathcal{M}_{st}$ with respect to $\sigma$. 
\end{theorem}

We also use the following known eigenvalue characterization of $\chi^{2}_{f}$ SDPI constants.
\begin{lemma}\cite{hiai2016contraction, cao2019tensorization,george2025unified} \label{lemma:commute}
Let $f \in \mathcal{M}_{st}$, and $\mathcal{E}$ be a quantum channel, then we have 
    \begin{equation}
        \eta_{\chi^2_{f}}(\mathcal{E},\sigma) = \lambda_2((\Omega_{\sigma}^{f})^{-1} \circ \mathcal{E}^{*} \circ \Omega_{\sigma}^{f} \circ \mathcal{E} )   
    \end{equation}
    where $\lambda_2(\cdot)$ denotes the second largest eigenvalue of the operator. 
\end{lemma}

Now we give the proof for the tightness result. 
\begin{theorem}
    Let $f \in \mathcal{F} \cap \mathcal{F}_{\mathrm{Pin}}$. Let $\mathcal{E}$ be a primitive channel with fixed point $\pi$, and additionally assume $f$ is thrice continuously differentiable in an open neighborhood of $1$. If $\mathcal{E}$ satisfies $\kappa_f$ detailed balance with respect to $\pi$ and $\mathbb{D}_f$ is locally $\chi^2_{\kappa_f}$, then
    \begin{equation}
        \eta_f(\mathcal{E}^{\circ n}, \pi) \geq  \eta_{\chi^2_{\kappa_f}}(\mathcal{E},\pi)^n
    \end{equation}
\end{theorem}

\begin{proof}
    By our assumption on locality and Lemma \ref{lemma:local}, $\eta_{f}(\mathcal{E}^{\circ n},\pi) \geq \eta_{\chi^2_{\kappa_f}}(\mathcal{E}^{\circ n},\pi)$. So it remains to show We have by the above lemma that $\eta_{f}(\mathcal{E}^{\circ n},\pi) \geq \eta_{\chi^2_{\kappa_f}}(\mathcal{E}^{\circ n},\pi)$, so it remains to show that 
    \begin{equation}
        \eta_{\chi^2_{\kappa_f}}(\mathcal{E}^{\circ n},\pi) = \eta_{\chi^2_{\kappa_f}}(\mathcal{E},\pi)^n
    \end{equation}
    when $\mathcal{E}$ satisfies $\kappa_f$ detailed balance with respect to $\pi$. 

    By Lemma \ref{lemma:commute}, we have the following chain of equalities 
    \begin{equation}
        \begin{aligned}
            \eta_{\chi^2_{\kappa_f}}(\mathcal{E}^{\circ n},\pi) & =  \lambda_2((\Omega_{\pi}^{\kappa_f})^{-1} \circ (\mathcal{E}^{\circ n})^{*} \circ \Omega_{\pi}^{\kappa_f} \circ \mathcal{E}^{\circ n} )\\
            & = \lambda_2(\mathcal{E}^{\circ 2n})\\
            & = \lambda_2(\mathcal{E}^{\circ 2})^n\\
            & = \{\lambda_2((\Omega_{\pi}^{\kappa_f})^{-1} \circ \mathcal{E}^{*} \circ \Omega_{\pi}^{\kappa_f} \circ \mathcal{E})\}^n\\
            & = \eta_{\chi^2_{\kappa_f}}
            (\mathcal{E},\pi)^{n}\\ 
        \end{aligned}
    \end{equation}
    where the first equality is from Lemma \ref{lemma:commute}, the second uses the detailed balance equation \eqref{eq:DB-equation} $n$ times to commute the $(\Omega_{\pi}^{\kappa_f})^{-1}$ operator over, the third equality is discussed below, the fourth again uses the detailed balance equation \eqref{eq:DB-equation}, and the last equality follows again from Lemma \ref{lemma:commute}. 
    
    The third equality is justified in the following manner. Since $\mathcal{E}$ is a quantum channel, we know that $\mathrm{spec}(\mathcal{E})$ is contained in the complex unit ball (see Chap. 6 of \cite{wolf2012quantum}). Since $\mathcal{E}$ satisfies $\kappa_f$ detailed balance with respect to $\pi$, it is self-adjoint with respect to the inner product defined by $(\Omega_{\pi}^{\kappa_f})^{-1}$, so $\mathrm{Spec(\mathcal{E})} \subset [-1,1]$. Therefore $\mathrm{Spec(\mathcal{E}^{\circ 2})} \subset [0,1]$. Thus the eigenvalue ordering is preserved when composing $n$ times and the third inequality holds. This completes the proof. 
\end{proof}

\section{Application to Various Quantum \texorpdfstring{$f$}{}-Divergences}\label{sec:applications}
In this section, we apply Theorem \ref{thm:main-result} to three families of quantum $f$-divergences for which the corresponding local $\chi^{2}_{g}$-divergences is already known: Petz, Hirche-Tomamichel (HT), and Matsumoto (which are maximal) $f$-divergences.

We begin by recalling the definitions of the different $f$-divergences. To that end, we introduce some preliminaries. As a Petz $f$-divergence requires $f$ to be operator convex to satisfy data processing \cite{hiai2011quantum,Hiai_2017}, we define the set
\begin{equation}
    \mathcal{F}_{\Petz} := \{ f \in \mathcal{F}:\text{ operator convex}\} \ . 
\end{equation}
The HT $f$-divergences are induced by the quantum hockey stick divergences, which we also recall.
\begin{definition}\cite{sharma2012strong}
    Let $\rho,\sigma \in \mathcal{D(H)}$ and $\gamma \geq 1$. Then the quantum hockey-stick divergence is defined as 
    \begin{equation}
        E_{\gamma}(\rho  \Vert  \sigma) = \Tr[(\rho - \gamma\sigma)_+]
    \end{equation}
    where $A_+$ denotes the positive part of the Hermitian operator $A$ in the Jordan decomposition. We set $E_\gamma$ to be infinity when $\rho$ is not absolutely continuous with respect to $\sigma$. 
\end{definition}

We now define the three families of quantum $f$-divergences.
\begin{definition}
    Let $f \in \cF$, $\mathcal{H}$ be a Hilbert space, and $\rho,\sigma \in \mathcal{D(H)}$.
    \begin{enumerate}[leftmargin=*]
        \item \cite{hirche2024quantum} The HT $f$-divergence is defined as 
        \begin{equation*}
            \hspace{-4mm} D_f(\rho  \Vert  \sigma) = \int_{1}^{\infty} f''(\gamma) E_{\gamma}(\rho \| \sigma) + \gamma^{-3} f''(\gamma^{-1}) E_{\gamma}(\sigma \| \rho) \, d\gamma 
        \end{equation*}
        when the integral is finite. Otherwise we set $D_f$ to infinity.
        \item \cite{petz1998contraction,matsumoto2015new} The Matsumoto $f$-divergence is defined as the following limit:
        \begin{align}
            \widehat{D}_{f}(\rho\Vert\sigma) = \lim_{\varepsilon \downarrow 0} \Tr[\sigma_{\ve} f(\sigma^{-1/2}_{\ve}\rho_{\ve}\sigma^{-1/2}_{\ve})] \ , 
        \end{align}
        where $\sigma_{\varepsilon} = \sigma + \varepsilon \mathds{1}$, $\rho_{\varepsilon}$ similarly, and $\mathds{1}$ is the identity on $\mathcal{H}$. 
        \item \cite{hiai2011quantum,Hiai_2017} If $f \in \cF_{\Petz}$, the Petz $f$-divergence is defined as the following limit:
        \begin{equation}
            \overline{D}_f(\rho\Vert\sigma) = \lim_{\varepsilon \downarrow 0}\Tr[\sigma_{\varepsilon}^{1/2}f(\Delta_{\rho_{\ve},\sigma_{\varepsilon}})\sigma_{\varepsilon}^{1/2}] \ .
        \end{equation}
        
    \end{enumerate}
\end{definition}

Next, we recall how these $f$-divergences behave locally.
\begin{prop}
    Let $\rho,\sigma \in \mathcal{D(H)}$ be full rank. For $f \in \cF$,
    \begin{enumerate}[leftmargin=*]
        \item \cite{hirche2024quantum} $\lim_{\lambda \downarrow 0} \frac{1}{\lambda^2}D_{f}(\lambda \rho + (1-\lambda)\sigma \Vert  \sigma) = \frac{f''(1)}{2}\chi^{2}_{\frac{\log(t)}{t-1}}(\rho \Vert \sigma)$
        \item \cite{matsumoto2015new} $\lim_{\lambda \downarrow 0} \frac{1}{\lambda^{2}}\widehat{D}_{f} (\lambda \rho + (1-\lambda)\sigma \Vert  \sigma) = \frac{f''(1)}{2}\chi^{2}_{\max}(\rho \Vert \sigma)$
        \item \cite{iyer2025quantum} if $f \in \cF_{\Petz}$ ,
    \begin{equation*}
        \lim_{\lambda \downarrow 0} \frac{1}{\lambda^2}\overline{D}_{f}(\lambda \rho + (1-\lambda)\sigma \Vert  \sigma) = \frac{f''(1)}{2} 
    \chi^2_{k_f}(\rho\Vert \sigma) \ ,
    \end{equation*}
    where $\kappa_f(x) = \frac{f(x) + xf(\frac{1}{x})}{f''(1)(x-1)^2}$.
    \end{enumerate}
\end{prop}

Thus, by applying Theorem \ref{thm:main-result} we obtain the following.
\begin{theorem}
    Let $f \in \mathcal{F}_{\mathrm{Pin}}, g \in \mathcal{M}_{st}$. Let $\mathcal{E}$ be a primitive channel with fixed point $\pi$, and additionally assume $f$ is thrice continuously differentiable in an open neighborhood of $1$. Then
    \begin{enumerate}[leftmargin=*]
        \item Rate for HT $f$-Divergences:
        \begin{equation}
            \lim_{n \to \infty}\eta_{D_{f}}(\mathcal{E}^{\circ n}, \pi)^{1/n} \leq \eta_{\chi^2_{g}}(\mathcal{E},\pi) \ , 
        \end{equation}
        and we have equality if $g(x) = \frac{log(x)}{x-1}$ and $\mathcal{E}$ satisfies $g$-detailed balance with respect to $\pi$.
        \item Rate for Matsumoto $f$-divergences:
        \begin{equation}
            \lim_{n \to \infty}\eta_{\widehat{D}_{f}}(\mathcal{E}^{\circ n}, \pi)^{1/n} \leq \eta_{\chi^2_{g}}(\mathcal{E},\pi) \ , 
        \end{equation}
        and we have equality if $g(x) = \frac{x+1}{2x}$ and $\mathcal{E}$ satisfies $g$-detailed balance with respect to $\pi$.
        \item Rate for Petz $f$-divergences: If it is also the case $f \in \cF_{\Petz}$,
        \begin{equation}
            \lim_{n \to \infty}\eta_{\overline{D}_{f}}(\mathcal{E}^{\circ n}, \pi)^{1/n} \leq \eta_{\chi^2_{g}}(\mathcal{E},\pi) \ . 
        \end{equation}
        and we have equality if $g(x) = \frac{f(x) + xf(\frac{1}{x})}{f''(1)(x-1)^2}$ and $\mathcal{E}$ satisfies $g$-detailed balance with respect to $\pi$. 
    \end{enumerate}
\end{theorem}
In the above theorem, Item 1 generalizes a result of \cite{george2025unified}, all cases of Item 2 are new, and Item 3 improves the upper bounds given in \cite{george2024divergence}, which were given by the $\chi^2_{\max}$ divergence. In all cases, this is the first result establishing tightness in the non-commutative setting. This highlights the general applicability of Theorem \ref{thm:main-result}.



\bibliographystyle{IEEEtran}
\bibliography{References.bib}
\appendix
In this section we provide proofs omitted from the main text. 

\begin{proof}[Proof of Lemma \ref{lemma:uniform}]
    We first prove that there is a sufficiently large $n$ such that $\mathcal{E}^{n}$ strictly contracts $\rho-\sigma$ for any $\rho,\sigma \in \mathcal{D(H)}$. 
    \begin{align}
        \Vert \cE^{\circ n}(\rho - \sigma) \Vert_{\infty} &\leq \frac{1}{2}\Vert \cE^{\circ n}(\rho - \sigma) \Vert_{1} \\
        &\leq \frac{1}{2}\Vert \rho - \sigma \Vert_{1} \eta_{\text{TD}}(\cE^{\circ n}) \\
        &\leq \eta_{\text{TD}}(\cE^{\circ n}) \\ 
        &\leq 1 - \alpha(\cE^{\circ n}) \ , 
    \end{align}
    where the first inequality is \cite[Proposition 78]{george2025unified}, the second is by definition of contraction coefficient for the trace norm, and the fourth is \cite[Proposition 28]{george2025quantum}. Here, $\alpha(\mathcal{E)}$ denotes the quantum Doeblin coefficient of $\mathcal{E}$ - see \cite{george2025quantum} for a definition. Now, by \cite[Theorem 6.8, Item 4]{wolf2012quantum} and \cite[Theorem 6.9]{wolf2012quantum}, a primitive channel $\cE: \mathcal{B}(\mbb{C}^{d}) \to \mathcal{B}(\mbb{C}^{d})$ applied $n \geq d^{4}$ times has a full rank Choi operator. By \cite[Proposition 36]{george2025quantum} and the previous observation, for $n \geq d^{4}$, $\alpha(\cE^{\circ n}) > 0$. Thus, for $n \geq d^{4}$, $\Vert \cE^{\circ n}(\rho - \sigma) \Vert_{\infty} = C < 1$. Moreover, by the data-processing inequality for the Schatten $1$-norm, for all $n' \geq n$, $\Vert \cE^{\circ n'}(\rho - \sigma) \Vert_{\infty} \leq C<1$. Iteratively applying this construction and using Markovianity finishes the proof. 
\end{proof}

\begin{proof}[Proof of Lemma \ref{lemma:eigenvalues}]
    \begin{equation}
      \begin{aligned}
          \norm{\mixedterm - \mathds{1}}_{\infty} & \leq \norm{\sigma^{-\frac{1}{2}}(\rho - \sigma)\sigma^{-\frac{1}{2}}}_{\infty} \\ 
          & \leq \norm{\sigma^{-\frac{1}{2}}}_{\infty}^{2}\norm{\rho-\sigma}_{\infty}\\ 
          & \leq \frac{1}{\lambda_{\min}(\sigma)}\norm{\rho - \sigma}_{\infty}
      \end{aligned}
    \end{equation}
     where we use the sub-multiplicativity of the $\infty$-norm twice. 
\end{proof}

\begin{proof}[Proof of Proposition \ref{lowerbound}]
   Let $X = \rho-\sigma$ and $\mu_i$ be the eigenvalues of $\sigma$. We have 
    \begin{equation}
        \begin{aligned}
             \chi^2_g\left(\rho  \Vert  \sigma\right) & = \sum_{ij}\frac{1}{\mu_j}g\left(\frac{\mu_i}{\mu_j}\right)|X_{ij}|^2 \\ 
          &  \leq   \max_{ij} \Big\{\frac{1}{\mu_j}g\left(\frac{\mu_i}{\mu_j}\right) \Big\} \Tr[X^{\dagger}X] \\
          & \leq   \max_{ij} \Big\{\frac{1}{\mu_j}g\left(\frac{\mu_i}{\mu_j}\right) \Big\} \norm{\rho - \sigma}_{2}^2 \\ 
          & \lesssim \norm{\rho - \sigma}_{1}^{2} \\
          & \lesssim_f \mathbb{D}_f(\rho  \Vert  \sigma)
        \end{aligned}
    \end{equation}
    where the last step references corollary 62 of \cite{george2024divergence}, and the $\max$ we take yields a strictly positive number since $g(x)>0$ for $x>0$ as $g \in \mathcal{M}_{st}$.\\

    The proof of the second statement follows from 
    \begin{equation}
        \begin{aligned}
            \chi^2_g(\rho  \Vert  \sigma) & = \sum_{ij}\frac{1}{\mu_j}g\left(\frac{\mu_i}{\mu_j}\right)|X_{ij}|^2 \\ 
          &  \geq  \min_{ij} \Big\{\frac{1}{\mu_j}g\left(\frac{\mu_i}{\mu_j}\right) \Big\} \Tr[X^{\dagger}X] \\
          & =  \min_{ij} \Big\{\frac{1}{\mu_j}g\left(\frac{\mu_i}{\mu_j}\right) \Big\} \norm{\rho - \sigma}_{2}^2
        \end{aligned}
    \end{equation}
    and we can finish by applying the equivalence of norms in finite dimensions. 
\end{proof}

\begin{proof}[Proof of \ref{lemma:local}]
    We imitate the proof in \cite{hirche2024quantum} for the HT-divergences. By definition we have 
    \begin{equation}
        \begin{aligned}
        \frac{1}{\lambda^2}D_{f}(\mathcal{E}(\lambda\rho+(1-\lambda)\sigma \Vert \mathcal{E}(\sigma)) & \\ \leq \eta_f(\mathcal{E},\sigma)\frac{1}{\lambda^2}D_{f}(\lambda \rho + (1-\lambda)\sigma \Vert \sigma)
    \end{aligned}
    \end{equation}
    Sending $\lambda \to 0$ and making use of the locality result, we have 
    \begin{equation}
    \frac{\chi^2_{\kappa_f}(\mathcal{E}(\rho) \Vert \mathcal{E}(\sigma))}{\chi^2_{\kappa_f}(\rho \Vert \sigma)} \leq \eta_f(\mathcal{E},\sigma)
    \end{equation}
    Now taking the supremum over the left hand side we have the result. 
\end{proof}

\end{document}